\documentclass{article}

\usepackage[affil-it]{authblk}
\usepackage[usenames,dvipsnames]{xcolor}
\usepackage{amsfonts}
\usepackage{amsmath,amsthm,amssymb,dsfont}
\usepackage{enumerate}
\usepackage{graphicx}	
\usepackage{subcaption}
\usepackage[margin=3cm]{geometry}
\usepackage{url}
\usepackage{todonotes}
\usepackage{bbm}

\usepackage{tikz}
\usetikzlibrary{chains}
\usetikzlibrary{fit}
\usepackage{pgflibraryarrows}		
\usepackage{pgflibrarysnakes}		

\usepackage{makecell}

\usepackage{epsfig}
\usetikzlibrary{shapes.symbols,patterns} 
\usepackage{pgfplots}

\usepackage{hyperref}
\hypersetup{colorlinks=true,citecolor=blue,linkcolor=blue,filecolor=blue,urlcolor=blue,breaklinks=true}

\usepackage{nicefrac}
\usepackage{mathtools}

\theoremstyle{plain}
\newtheorem{theorem}{Theorem}[section]
\newtheorem{lemma}[theorem]{Lemma}

\newtheorem{corollary}[theorem]{Corollary}
\newtheorem{proposition}[theorem]{Proposition}

\theoremstyle{definition}

\newtheorem{remark}[theorem]{Remark}

\newcommand*{\cA}{\mathcal{A}}
\newcommand*{\cB}{\mathcal{B}}

\newcommand*{\cE}{\mathcal{E}}
\newcommand*{\cF}{\mathcal{F}}
\newcommand*{\cG}{\mathcal{G}}

\newcommand*{\cI}{\mathcal{I}}

\newcommand*{\cN}{\mathcal{N}}

\newcommand*{\cR}{\mathcal{R}}

\newcommand*{\cX}{\mathcal{X}}
\newcommand*{\cY}{\mathcal{Y}}

\newcommand*{\N}{\mathbb{N}}

\newcommand*{\R}{\mathbb{R}}

\newcommand*{\St}{\mathrm{S}}

\newcommand*{\eps}{\varepsilon}
\newcommand*{\diag}{\mathrm{diag}}

\newcommand*{\id}{\mathrm{id}}

\newcommand*{\supp}{\mathrm{supp}}
\newcommand*{\tr}{\mathrm{tr}}
\newcommand*{\ket}[1]{| #1 \rangle}
\newcommand*{\bra}[1]{\langle #1 |}

\newcommand{\proj}[1]{|#1\rangle\!\langle #1|}

\newcommand*{\Pos}{\mathrm{P}}
\newcommand*{\BB}{\mathrm{B}}

\newcommand*{\TPCP}{\mathrm{TPCP}}
\newcommand*{\CP}{\mathrm{CP}}

\newcommand{\norm}[1]{\left\lVert#1\right\rVert}

\newcommand*{\ox}{\otimes}

\usepackage{float}

\title{
A chain rule for the quantum relative entropy}

\author[1]{Kun Fang}
\author[2]{Omar Fawzi}
\author[3]{Renato Renner}
\author[3]{David Sutter}

\affil[1]{\small{Department of Applied Mathematics and Theoretical Physics, University of Cambridge, UK}}
\affil[2]{\small{Laboratoire de l'Informatique du Parall\'elisme, ENS de Lyon, France}}
\affil[3]{\small{Institute for Theoretical Physics, ETH Zurich, Switzerland}}
\date{}

\begin{document}

\maketitle

\begin{abstract}
The chain rule for the classical relative entropy ensures that the relative entropy between probability distributions on multipartite systems can be decomposed into a sum of relative entropies of suitably chosen conditional distributions on the individual systems. Here, we prove a similar chain rule inequality for the quantum relative entropy in terms of channel relative entropies. The new chain rule allows us to solve an open problem in the context of asymptotic quantum channel discrimination: surprisingly, adaptive protocols cannot improve the error rate for asymmetric channel discrimination compared to non-adaptive strategies. In addition, we give examples of quantum channels showing that the channel relative entropy is not additive under the tensor product.
\end{abstract}

\section{Introduction}
The quantum \emph{relative entropy} between a state $\rho$ and a positive semidefinite operator $\sigma$ defined as
\begin{align*}
D(\rho \| \sigma):=\left \lbrace
\begin{array}{ll}
\tr\, \rho \,(\log \rho - \log \sigma) & \text{if } \supp(\rho) \subseteq \supp(\sigma) \\
+\infty & \text{otherwise} \, ,
\end{array} \right.
\end{align*}
is an important entropic quantity in physics. Many other information measures such as von Neumann entropies, conditional entropies, or mutual informations can be viewed as a relative entropy between suitably chosen operators. Hence it is crucial to understand the mathematical properties of this quantity. One aspect which remains poorly understood is how to decompose the relative entropy between states on multipartite systems into a sum of relative entropies between states on the individual systems. In the classical case, the well-known chain rule for relative entropy can be used~\cite[Theorem~2.5.3]{cover}. For a pair of discrete random variables $(X, Y)$ with alphabet $\cX \times \cY$, we have
\begin{align*}
D(P_{XY} \| Q_{XY}) = D(P_{X} \| Q_{X}) + \sum_{x \in \cX} P_{X}(x) D(P_{Y|X=x} \| Q_{Y|X=x}) \ ,
\end{align*}
where $P_{XY}$ and $Q_{XY}$ are joint probability distributions, but $Q_{XY}$ does not need to be normalized. No quantum analogue of such a chain rule is known, even if we relax the equality with the following inequalities
\begin{align}
D(P_{X} \| Q_{X}) + \min_{x \in \cX} D(P_{Y|X=x} \| Q_{Y|X=x}) 
&\leq D(P_{XY} \| Q_{XY}) \nonumber \\ 
&\leq D(P_{X} \| Q_{X}) + \max_{x \in \cX} D(P_{Y|X=x} \| Q_{Y|X=x}) \ .\label{eq_chainrule_classical}
\end{align}
In this manuscript, we prove a quantum version of the upper bound~\eqref{eq_chainrule_classical} and show that it is tight in the sense that there exist scenarios where the chain rule is an equality.

To model the quantum setting, the conditional distributions are replaced by trace-preserving completely positive maps $\cE$ and $\cF$ from $A$ to $B$ and the initial states are density operators $\rho_{RA}$ and $\sigma_{RA}$, where $R$ denotes a reference system. To express the very last term in~\eqref{eq_chainrule_classical} in the quantum mechanical case we use a quantity called (stabilized) \emph{channel relative entropy}, which is defined by
\begin{align} \label{eq_channel_rel_ent} 
D(\cE \| \cF) := \max_{\phi_{RA} \in \St(A \otimes A) } D\big( (\cI \otimes \cE)(\phi_{RA} ) \| (\cI \otimes \cF)(\phi_{RA} ) \big) \, ,
\end{align}
where $\St(A)$ denotes the set of density operators on $A$ and $\cI$ denotes the identity map.
Motivated by the classical case~\eqref{eq_chainrule_classical}, it is natural to ask whether the following chain rule is correct
\begin{align} \label{eq_hope}
D\big( (\cI \otimes \cE)(\rho_{RA}) \| (\cI \otimes \cF)(\sigma_{RA})\big)  \stackrel{?}{\leq}  D( \rho_{RA} \| \sigma_{RA}) + D(\cE \| \cF) \ .
\end{align}
It turns out that this does not hold in general as we show in Proposition~\ref{prop_notAdditive}. By regularizing the channel relative entropy term, the inequality however becomes valid, i.e., for $D^{\mathrm{reg}}(\cE \| \cF):=  \lim_{n \to \infty} \frac{1}{n}D(\cE^{\otimes n} \| \cF^{\otimes n})$ the inequality
\begin{align} \label{eq_resIntro}
D\big( (\cI \otimes \cE)(\rho_{RA}) \| (\cI \otimes \cF)(\sigma_{RA}) \big)  \leq  D( \rho_{RA} \| \sigma_{RA}) + D^{\mathrm{reg}}(\cE \| \cF)
\end{align}
is correct.
We refer to Theorem~\ref{thm_chainRule} and Corollary~\ref{corollary_collapse} for a more precise and a more general result. 

In Section~\ref{sec_chainRule} we discuss the limitations of a chain rule of the form~\eqref{eq_hope} and, based on smooth entropy calculus, derive the chain rule of the form~\eqref{eq_resIntro}. We also show that for any two trace-preserving completely positive maps $\cE$ and $\cF$ there exist states $\rho_{RA}$ and $\sigma_{RA}$ such that~\eqref{eq_resIntro} holds with equality. 
This answers an open question in the area of channel discrimination showing that, surprisingly, adaptive strategies cannot be more powerful than non-adaptive strategies~\cite{hayashi2009discrimination,berta18,Wang2019,banff19}. We comment on this in Section~\ref{sec_channelDiscrimination}.

\section{Preliminaries}
The set of positive semidefinite operators on $A$ is denoted by $\Pos(A)$, the set of density operators is given by $\St(A)$, and the set of positive semidefinite operators with trace at most one is denoted by $\St_{\leq}(A)$.
The set of completely positive and trace-preserving completely positive maps from linear operators on $A$ to linear operators on $B$ is denoted by $\CP(A,B)$ and $\TPCP(A,B)$, respectively. We drop identity maps if they are clear from the context. For example for $\rho_{RA} \in \St(R\otimes A)$ and $\cE \in \TPCP(A,B)$ we write $\cE(\rho_{RA})$ instead of $(\cI \otimes \cE)(\rho_{RA})$.
For $\cE \in \CP(A,B)$ we denote its \emph{Choi state} by $J^{\cE}_{RB} :=  \cE (\proj{\Omega}_{RA})$, where $\ket{\Omega}_{RA} = \sum_{i} \ket{i}_{R} \ket{i}_A$ is the unnormalized maximally entangled state.
For $\rho, \sigma \in \Pos(A)$ we define the \emph{fidelity} by $F(\rho,\sigma):=\|\sqrt{\rho} \sqrt{\sigma}\|_1^2$, where $\norm{X}_1:= \tr (X^\dagger X)^{1/2}$ is the trace norm of $X$.  
For $\rho, \sigma \in \St_{\leq}(A)$ the \emph{purified distance} is given by $P(\rho,\sigma):=\sqrt{1-F^{\star}(\rho,\sigma)}$, where $F^{\star}(\rho,\sigma):=(\norm{\sqrt{\rho} \sqrt{\sigma}}_1+\sqrt{(1-\tr\, \rho)(1-\tr\, \sigma)})^2$ is the \emph{generalized fidelity}~\cite{TCR10}.
For  $\rho \in \St_{\leq}(A)$ and  $\eps \in [0, \sqrt{\tr\, \rho})$ the $\eps$-ball around $\rho$ is given by $\BB_{\eps}(\rho):=\{\omega \in \St_{\leq} (A): P(\omega,\rho) \leq \eps\}$.
The \emph{max-relative entropy} for $\rho,\sigma \in \Pos(A)$ is defined as~\cite{renner_phd,datta09}
\begin{align*}
D_{\max}(\rho \| \sigma):= \inf \{\lambda \in \R : \rho \leq 2^{\lambda} \sigma\} \, .
\end{align*}
Its smooth version is given by $D_{\max}^{\eps}(\rho \| \sigma):= \inf_{\tilde \rho \in \BB_{\eps}(\rho)} D_{\max}(\tilde \rho \| \sigma)$. 

For $\cE \in \TPCP(A,B)$ and $\cF \in \CP(A,B)$ we defined the (stabilized) \emph{channel relative entropy} $D(\cE \| \cF) $ in~\eqref{eq_channel_rel_ent}.
Its non-stabilized counterpart is given by
\begin{align*}
\bar D(\cE \| \cF) := \max_{\phi_{A} \in \St(A) }D\big(  \cE(\phi_A) \| \cF(\phi_A) \big) \, .
\end{align*}
The regularized version is given by $\bar D^{\mathrm{reg}}(\cE \| \cF):= \lim_{n \to \infty} \frac{1}{n} \bar D(\cE^{\otimes n} \| \cF^{\otimes n})$.
By definition we have $\bar D(\cE \| \cF) \leq D(\cE \| \cF)$. The (stabilized) channel max-relative entropy accordingly is defined as
\begin{align*}
D_{\max}(\cE \| \cF) := \max_{\phi_{RA} \in \St(A \otimes A) } D_{\max}\big( \cE(\phi_{RA}) \|  \cF(\phi_{RA}) ) \, \, .
\end{align*}
The channel max-relative entropy can be expressed in a simple closed form as a function of the Choi states of the maps $\cE$ and $\cF$~\cite[Lemma~12]{berta18}. This implies that it is also additive under tensor products, i.e., $D_{\max}(\cE^{\otimes n} \| \cF^{\otimes n}) = n D_{\max}(\cE \| \cF)$ for all $n \in \N$.
Another quantity that is of interest is the \emph{amortized channel relative entropy}~\cite{berta18} defined by
\begin{align}  \label{eq_amortized}
D^{A}(\cE \| \cF) := \sup_{\phi_{RA}, \omega_{R A} \in \St(R\, \otimes A)} \{  D\big( \cE(\phi_{RA}) \| \cF(\omega_{RA})\big) -  D( \phi_{RA} \|\omega_{RA}) \} \, .
\end{align}
We note that unlike for the channel relative entropy defined in~\eqref{eq_channel_rel_ent}, where the reference system $R$ can be assumed to be isomorphic to the input system $A$, it is unclear if this assumption can made for the amortized channel relative entropy. Here we do not constrain the dimension of the $R$ system which makes the amortized channel divergence a more complicated quantity.

%

\section{Chain rule for the quantum relative entropy} \label{sec_chainRule}
In this section we prove our main result which is a chain rule for the relative entropy of completely positive maps. Furthermore, we comment on possible generalizations and limitations of such chain rules. Finally we discuss its implications, in particular for the task of asymptotic quantum channel discrimination.

\subsection{Non-additivity of the channel relative entropy}
As mentioned above the channel max-relative entropy is additive under tensor products. Here we show that this property is not true for the channel relative entropy.
This then implies that the naive guess for a possible chain rule for the relative entropy mentioned in~\eqref{eq_hope} is not correct.
\begin{proposition}[Channel relative entropy is not additive under tensor product] \label{prop_notAdditive}
There exist $\cE,\cF \in \TPCP(A,B)$ such that 
\begin{align} \label{eq_notAdditive}
D(\cE \otimes \cE \| \cF \otimes \cF) > 2 D(\cE \| \cF) \, .
\end{align}
This implies that there exist $\rho_{RA}, \sigma_{RA} \in \St(R \otimes A)$ for some finite-dimensional system $R$ such that
\begin{align} \label{eq_naiveChainRuleWrong}
D\big( \cE(\rho_{RA}) \| \cF(\sigma_{RA}) \big) > D( \rho_{RA} \| \sigma_{RA}) + D(\cE \| \cF) \ .
\end{align}
\end{proposition}
\begin{proof}
We start by proving that~\eqref{eq_notAdditive} implies~\eqref{eq_naiveChainRuleWrong}. It is known~\cite[Theorem 3 and 6]{Wang2019} (see Section~\ref{sec_channelDiscrimination} for more explanations) that
\begin{align}
	D(\cE\|\cF) \leq D^{\mathrm{reg}}(\cE\|\cF)  \leq D^{A}(\cE\|\cF)\, .
\end{align}
The statement~\eqref{eq_notAdditive} implies that the first inequality can be strict. By definition of the amortized channel relative entropy~\eqref{eq_amortized} this directly implies~\eqref{eq_naiveChainRuleWrong}.

It thus remains to prove~\eqref{eq_notAdditive}. To do so we construct an example of two trace-preserving completely positive maps $\cE$ and $\cF$ on qubits that satisfy~\eqref{eq_notAdditive}. 
Consider the generalized amplitude damping channel 
\begin{align*}
	\cA_{\gamma,\beta} (\rho) = \sum_{i=1}^4 A_i \rho A_i^\dagger, \quad \text{for } \gamma, \beta \in [0,1]
\end{align*}
with the Kraus operators
\begin{alignat*}{2}
	& A_1  = \sqrt{1-\beta} (\ket{0}\bra{0} + \sqrt{1-\gamma}\ket{1}\bra{1})\, , \quad && A_2 = \sqrt{\gamma(1-\beta)}\ket{0}\bra{1}\, ,\\
& A_3 = \sqrt{\beta} (\sqrt{1-\gamma}\ket{0}\bra{0}+\ket{1}\bra{1})\, , && A_4 = \sqrt{\gamma \beta}\ket{1}\bra{0} \, .
\end{alignat*}
For the two channels $\cE = \cA_{0.3,0}$ and $\cF = \cA_{0.5,0.9}$ their corresponding Choi matrices are given by
\begin{align*}
 	J_{RB}^{\cE} = \begin{pmatrix}
 		1 & 0 & 0 & \sqrt{0.7}\\
 		0 & 0 & 0 & 0\\
 		0 & 0 & 0.3 & 0\\
 		\sqrt{0.7} & 0 & 0 & 0.7
 	\end{pmatrix}\quad \text{and} \quad J_{RB}^{\cF} = 
 	\begin{pmatrix}
 		0.55 & 0 & 0 & \sqrt{0.5}\\
 		0 & 0.45 & 0 & 0\\
 		0 & 0 & 0.05 & 0\\
 		\sqrt{0.5} & 0 & 0 & 0.95
 	\end{pmatrix}.
 \end{align*}
 For an arbitrary density matrix $\rho \in \St(A)$ let $\ket{\phi}_{RA}=(\sqrt{\rho_R} \otimes \id_A) \ket{\Omega}_{RA}$ be its purification where $\ket{\Omega}_{RA}= \sum_i \ket{i}_{R} \ket{i}_A$ and where $R$ is isomorphic to $A$. Hence we find
\begin{align} \label{eq_kun1}
\cE(\phi_{RA}) 
=(\cI \otimes \cE)\Big((\sqrt{\rho_R} \otimes \id_A) \proj{\Omega}_{RA} (\sqrt{\rho_R} \otimes \id_A )\Big)
= \sqrt{\rho_R} J_{RB}^{\cE}\sqrt{\rho_R} \, .
\end{align}
Using~\eqref{eq_kun1} gives
 \begin{align*}
 	D(\cE\|\cF)  
	= \max_{\rho_R \in \St(R)} D\big(\sqrt{\rho_R} J_{RB}^{\cE} \sqrt{\rho_R}\|\sqrt{\rho_R} J_{RB}^{\cF} \sqrt{\rho_R}\big)
	 =\!\! \max_{\rho_R = \diag(p,1-p)} \!\!D\big(\sqrt{\rho_R} J_{RB}^{\cE} \sqrt{\rho_R}\|\sqrt{\rho_R} J_{RB}^{\cF} \sqrt{\rho_R}\big).
 \end{align*}
The final step follows since both $\cE$ and $\cF$ are covariant with respect to the Pauli-$Z$ operator. Thus it suffices to perform the maximization over input states with respect to the one-parameter family of states $\rho_R = p \ket{0}\bra{0} + (1-p)\ket{1}\bra{1}$ (see e.g.~\cite[Proposition II.4.]{leditzky2018approaches}). Using the \texttt{fminbnd} function in Matlab, we find $D(\cE\|\cF) = 0.9176$ for an optimizer $\rho_R = \diag(0.8355, 1-0.8355)$. 
 This can also be seen by plotting the value of the relative entropy over the interval $p \in [0,1]$ as shown in Figure~\ref{nonadditivity plots}.

\begin{figure}[H]
\centering
\begin{tikzpicture}
\begin{scope}[shift={(-4,0)},scale = 1.1]
\node at (0,0) {\includegraphics[height = 4.9cm]{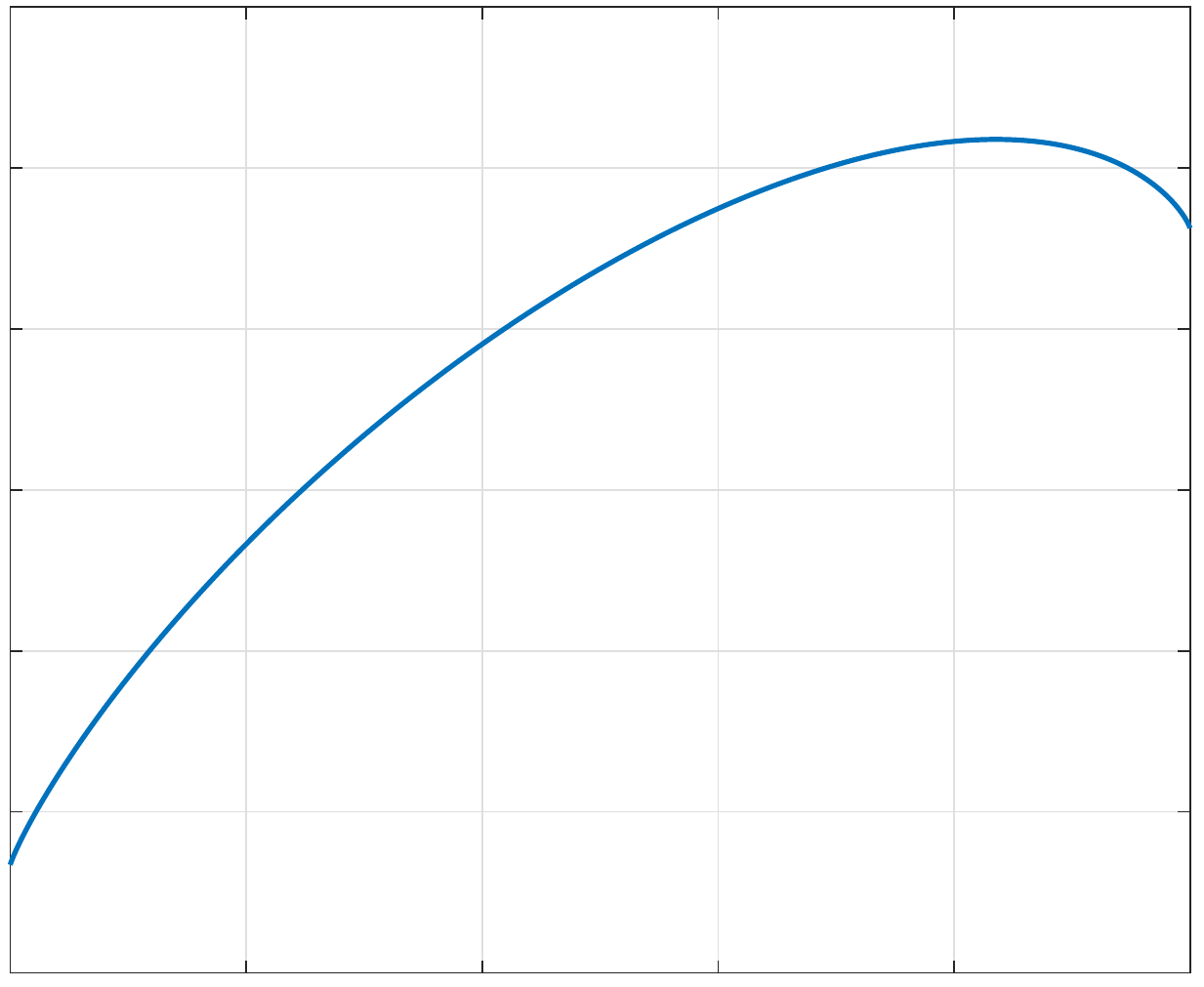}};
\node at (-2.65,-2.4) {\small $0$};
\node at (-1.55,-2.4) {\small $0.2$};
\node at (-0.5,-2.4) {\small $0.4$};
\node at (0.5,-2.4) {\small $0.6$};
\node at (1.55,-2.4) {\small $0.8$};
\node at (2.65,-2.4) {\small $1$};
\node at (-3,-2.1) {\small $0.4$};
\node at (-3,-1.4) {\small $0.5$};
\node at (-3,-0.7) {\small $0.6$};
\node at (-3,0) {\small $0.7$};
\node at (-3,0.7) {\small $0.8$};
\node at (-3,1.4) {\small $0.9$};
\node at (-3,2.1) {\small $1$};
\node at (1,-0.5) {\includegraphics[height = 1.8cm]{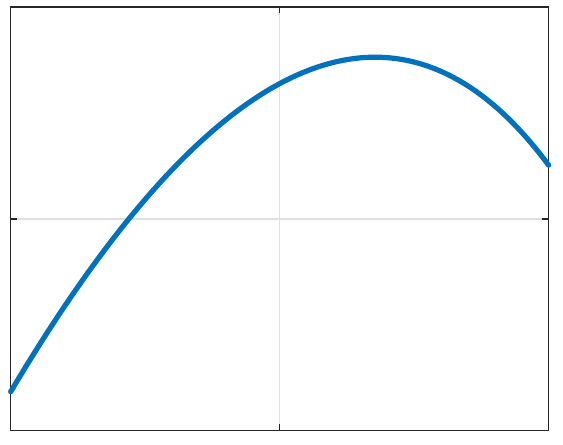}};
\node at (0.1,-1.5) {\footnotesize $0.7$};
\node at (1,-1.5) {\footnotesize $0.8$};
\node at (1.9,-1.5) {\footnotesize $0.9$};
\node at (-0.3,-1.2) {\footnotesize $0.9$};
\node at (-0.3,-0.5) {\footnotesize $0.91$};
\node at (-0.3,0.2) {\footnotesize $0.92$};

\node at (0,-2.8) {$p$};

\end{scope}
\begin{scope}[shift={(4,0)},scale=1.1]
\node at (0,0) {\includegraphics[height = 5cm]{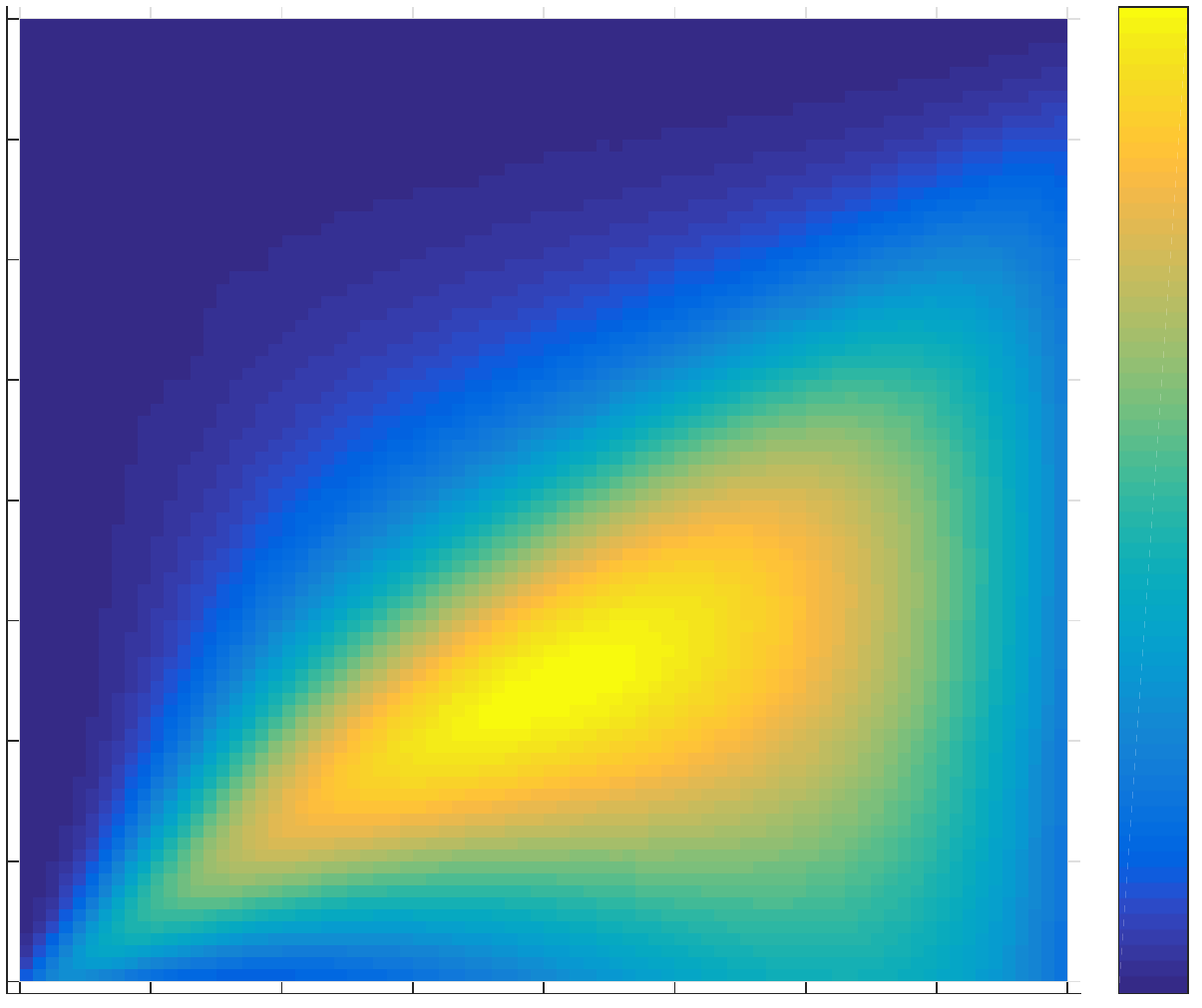}};
\node at (-2.6,-2.4) {\small $0.1$};
\node at (-2,-2.4) {\small $0.2$};
\node at (-1.4,-2.4) {\small $0.3$};
\node at (-0.8,-2.4) {\small $0.4$};
\node at (-0.2,-2.4) {\small $0.5$};
\node at (0.4,-2.4) {\small $0.6$};
\node at (1,-2.4) {\small $0.7$};
\node at (1.6,-2.4) {\small $0.8$};
\node at (2.2,-2.4) {\small $0.9$};

\node at (-3,-2.2) {\small $0.1$};
\node at (-3,-1.6) {\small $0.2$};
\node at (-3,-1.05) {\small $0.3$};
\node at (-3,-0.5) {\small $0.4$};
\node at (-3,0) {\small $0.5$};
\node at (-3,0.55) {\small $0.6$};
\node at (-3,1.05) {\small $0.7$};
\node at (-3,1.62) {\small $0.8$};
\node at (-3,2.15) {\small $0.9$};

\node at (3,-2.1) {\small $0$};
\node at (3,-1.4) {\small $0.02$};
\node at (3,-0.6) {\small $0.04$};
\node at (3,0.2) {\small $0.06$};
\node at (3,1) {\small $0.08$};
\node at (3,1.8) {\small $0.1$};

\node at (0,-2.8) {$\gamma_2$};
\node at (-3.5,0) {$\gamma_1$};
\end{scope}
\end{tikzpicture}
\caption{\small The left figure plots of the value $D(\sqrt{\rho_R} J_{RB}^{\cE} \sqrt{\rho_R}\|\sqrt{\rho_R} J_{RB}^{\cF} \sqrt{\rho_R})$ with respect to the input state $\rho_R = \diag(p,1-p)$. The subfigure is a zoom in plot with the parameter ranging from $0.7$ to $0.9$. It is evident that $D(\cE\|\cF)$ cannot be larger than $0.92$. The right figure shows a heat map of the value $D(\cA_{\gamma_1,0}^{\ox 2}\|\cA_{\gamma_2,0.9}^{\ox 2}) - 2 D(\cA_{\gamma_1,0}\|\cA_{\gamma_2,0.9})$ where $\gamma_1, \gamma_2 \in [0.1,0.9]$. This shows that the non-additivity behavior of the channel relative entropy under tensor products occurs for many channels.
}
\label{nonadditivity plots}
\end{figure}
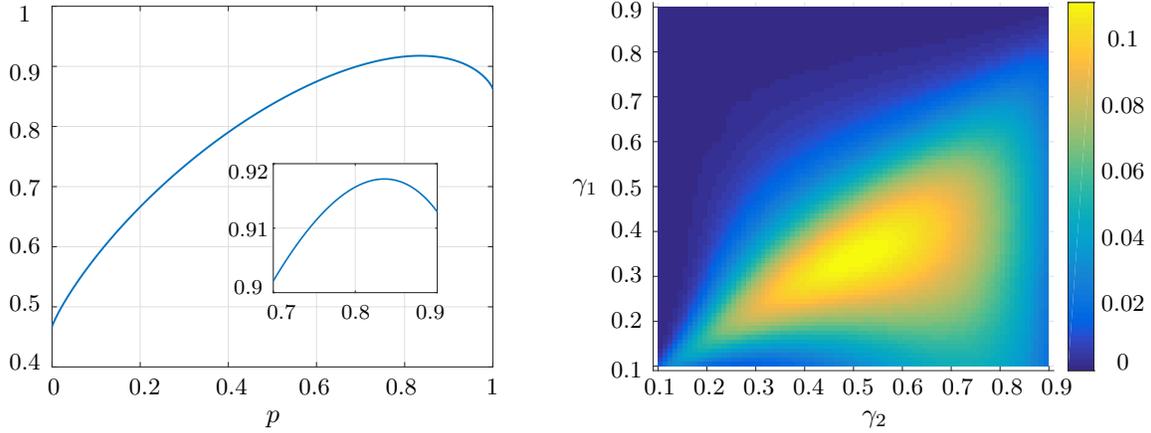

On the other hand, if we choose the input state $\rho_{R_1R_2} = \diag(0.8,0,0,0.2)$ we have
\begin{align*}
D(\cE \otimes \cE \| \cF \otimes \cF)
	& \geq D\big(\sqrt{\rho_{R_1R_2}} (J_{RB}^{\cE})^{\ox 2} \sqrt{\rho_{R_1R_2}}\|\sqrt{\rho_{R_1 R_2}} (J_{RB}^{\cF})^{\ox 2} \sqrt{\rho_{R_1 R_2}}\big)\\
	& = 1.9362 \\
	& > 2 \times 0.92\\
	& > 2 D(\cE\|\cF)\,.
\end{align*}
The difference is $1.9362 - 2 \times 0.9176 = 0.1010$ and shows that the channel relative entropy is not additive. 
More generally, as shown in Figure~\ref{nonadditivity plots}, we can plot the difference $D(\cA_{\gamma_1,0}^{\ox 2}\|\cA_{\gamma_2,0.9}^{\ox 2}) - 2 D(\cA_{\gamma_1,0}\|\cA_{\gamma_2,0.9})$ for a wide range of $\gamma_1, \gamma_2$. Due to the symmetry of the channel $\cA_{\gamma,\beta}^{\ox 2}$, we can restrict the computation of $D(\cA_{\gamma_1,0}^{\ox 2}\|\cA_{\gamma_2,0.9}^{\ox 2})$ to a two-parameter state $\rho_{R_1R_2} = \diag(p_1,p_2,p_2,1-p_1-2p_2)$ and we utilize the function \texttt{quantum\_rel\_entr} from \texttt{CVXQUAD}~\cite{cvxquad}. We observe that the relative entropy is not additive for a wide range of parameters.
\end{proof}

\subsection{Chain rule for smooth max-relative entropy}

The max-relative entropy has many desirable properties other relative entropies do not share. For example we can utilize the fact that it satisfies the triangle inequality (see e.g.~\cite[Lemma~2.1]{Sutter2018}) to immediately prove a chain rule of the form~\eqref{eq_hope}. In fact, one can even prove a strengthened chain rule as follows. Let $\rho \in \St(A)$, $\sigma \in \Pos(A)$, $\cE \in \TPCP(A,B)$, $\cF\in \CP(A,B)$, and $\cG \in \TPCP(A,C)$ such that there exists $\cR \in \CP(C,B)$ satisfying $\cF = \cR \circ \cG$.\footnote{Note that we can always choose $\cG = \cI$ and $\cR = \cF$.} Then
\begin{align}
D_{\max}\big(\cE(\rho) \| \cF(\sigma) \big) 
&\leq D_{\max}\big(\cE(\rho) \| \cF(\rho) \big) + D_{\max}\big(\cF(\rho) \| \cF(\sigma) \big) \nonumber \\
&=D_{\max}\big(\cE(\rho) \| \cF(\rho) \big) + D_{\max}\big(\cR \circ \cG (\rho) \| \cR \circ \cG (\sigma) \big) \nonumber \\
&\leq D_{\max}\big( \cG(\rho) \| \cG(\sigma) \big) +  D_{\max}\big(\cE(\rho) \|\cF(\rho) \big) \, , \label{eq_chainDmax}
\end{align}
where the final step uses the data-processing inequality for the max-relative entropy. We note that this fact has been observed also in~\cite[Section IV.D]{berta18}. The relative entropy does not satisfy the triangle inequality, so this simple proof does not work for the relative entropy. Instead our proof strategy will be to first show a chain rule for the smooth max-relative entropy and then apply the asymptotic equipartition property to derive the statement for the relative entropy.

In the following we will restrict the map $\cG$ to be a partial trace. More precisely we consider $A = A_1 \otimes A_2$ and $\cG = \tr_{A_1}$. We note that having a non-trivial map $\cG$ in the form of a partial trace in the chain rule turns out to be important to reproduce existing results as discussed in Remark~\ref{rmk_chainRuleCond}.
\begin{proposition}
\label{prop_dmax_chainrule}
Let $\eps, \eps' \in (0,1]$, $m \in \N$, $\rho_{A_1 A_2} \in \St(A_1 \otimes A_2)$, $\sigma_{A_1 A_2} \in \Pos(A_1 \otimes A_2)$, $\cE \in \TPCP({A_1 \otimes A_2},B)$, $\cF\in \CP(A_1 \otimes A_2,B)$ such that there exists $\cR \in \CP(A_2,B)$ satisfying $\cF = \cR \circ \tr_{A_1}$. Then
\begin{align} 
&D^{m\eps + \sqrt{m \eps} + \eps'}_{\max}\big(\cE(\rho_{A_1 A_2})^{\otimes m} \| \cF(\sigma_{A_1 A_2})^{\otimes m} \big)  \nonumber \\
&\hspace{20mm}\leq m D^{\eps}_{\max}\big(\rho_{A_2} \|\sigma_{A_2} \big) + \!\max_{\nu \in \St(A_1 \otimes A_2)}\! D_{\max}^{\eps'}\big(\cE(\nu)^{\otimes m} \| \cF(\nu)^{\otimes m} \big) - m \log(1-\eps) \, . \label{eq_dmax_chainrule}
\end{align}
\end{proposition}
\begin{proof}
There exists $\omega_{A_2} \in \BB_{\eps}(\rho_{A_2})$ such that 
\begin{align*}
\omega_{A_2} \leq 2^{D^{\eps}_{\max}(\rho_{A_2} \| \sigma_{A_2} )} \sigma_{A_2} \, .
\end{align*}
Using the property of the purified distance~\cite[Section~3.3]{tomamichel_phd}, there exists $\rho_{A_1 A_2}^{\eps} \in \BB_{\eps}(\rho_{A_1 A_2})$ such that $\omega_{A_2} = \tr_{A_1}\, \rho_{A_1 A_2}^{\eps}$. Note that we have $\tr\, \rho_{A_1 A_2} - \tr\, \rho_{A_1 A_2}^{\eps} \leq P(\rho_{A_1 A_2}, \rho_{A_1 A_2}^{\eps}) \leq \eps$ (see e.g., \cite[Lemma 3.5]{marco_book}) and thus 
\begin{align*}
\tr\, \rho_{A_1 A_2}^{\eps} \geq 1-\eps \, .
\end{align*}
Then, setting $\nu_{A_1A_2}^{\eps} =\rho_{A_1A_2}^{\eps}/ \tr \,\rho_{A_1A_2}^{\eps}$ we have that there exists $\tau_B \in \BB_{\eps'}(\cE(\nu_{A_1 A_2}^{\eps})^{\otimes m})$ such that 
\begin{align*}
\tau_B &\leq 2^{\max_{\nu \in \St(A_1 \otimes A_2)} D_{\max}^{\eps'}(\cE(\nu)^{\otimes m} \| \cF(\nu)^{\otimes m} )} \cF(\nu_{A_1 A_2}^{\eps})^{\otimes m} \\
&= 2^{\max_{\nu \in \St(A_1 \otimes A_2)} D_{\max}^{\eps'}(\cE(\nu)^{\otimes m} \| \cF(\nu)^{\otimes m} )} \frac{1}{\tr(\rho_{A_1A_2}^{\eps})^m} \cF(\rho_{A_1 A_2}^{\eps})^{\otimes m} \, .
\end{align*}
Now using the fact that $\cF = \cR \circ \tr_{A_1}$, we get
\begin{align*}
\tau_B 
&\leq 2^{\max_{\nu \in \St(A_1 \otimes A_2)} D_{\max}^{\eps'}(\cE(\nu)^{\otimes m} \| \cF(\nu)^{\otimes m} ) - m \log \tr(\rho_{A_1A_2}^{\eps}) } \cR^{\otimes m}\big( \tr_{A_1}( \rho_{A_1 A_2}^{\eps})^{\otimes m} \big) \\
& \leq 2^{\max_{\nu \in \St(A_1 \otimes A_2)} D_{\max}^{\eps'}(\cE(\nu)^{\otimes m} \| \cF(\nu)^{\otimes m} ) - m \log(1-\eps) + m D^{\eps}_{\max}( \rho_{A_2} \| \sigma_{A_2} )} \cF(\sigma_{A_1 A_2})^{\otimes m} \ ,
\end{align*}
where we used that the map $\cR$ is completely positive and the fact that $\tr\,\rho^{\eps} \geq 1-\eps$.
We now bound the purified distance using the triangle inequality
 \begin{align}
 &P\big(\tau_B, \cE(\rho_{A_1 A_2})^{\otimes m}\big) \nonumber \\
 &\hspace{10mm}\leq P\big(\tau_B, \cE(\nu_{A_1 A_2}^{\eps})^{\otimes m}\big) + P\big(\cE(\nu_{A_1 A_2}^{\eps})^{\otimes m}, \cE(\rho_{A_1 A_2}^{\eps})^{\otimes m}\big) + P\big(\cE(\rho_{A_1 A_2}^{\eps})^{\otimes m}, \cE(\rho_{A_1 A_2})^{\otimes m}\big) \ . \label{eq_distP}
 \end{align}
 By definition of $\tau$, we have $P\big(\tau_B, \cE(\nu_{A_1 A_2}^{\eps})^{\otimes m}\big) \leq \eps'$. For the second term we can use the definition of the purified distance to bound
 \begin{align*}
 P\big(\cE(\nu_{A_1 A_2}^{\eps})^{\otimes m}, \cE(\rho_{A_1 A_2}^{\eps})^{\otimes m}\big) 
 \!=\! \sqrt{1 - \left\| \frac{\cE(\rho_{A_1 A_2}^{\eps})^{\otimes m}}{\sqrt{\tr(\rho^{\eps})^{m} }} \right\|^2_1} 
 \!=\! \sqrt{1-(\tr\, \rho_{A_1 A_2}^\eps)^m}
 \!\leq\! \sqrt{1 - (1-\eps)^{m}} 
 \!\leq\! \sqrt{m \eps} \ ,
 \end{align*}
 where the final step uses the fact that $(1-\eps)^{m} \geq 1-m \eps$, which follows from the convexity of the function $ [0,1] \ni x \mapsto (1-x)^m$.
 For the last term in~\eqref{eq_distP}, we use the triangle inequality $m$ times together with the monotonicity of the purified distance under completely positive trace-preserving maps to get
 \begin{align*}
 P\big(\cE(\rho_{A_1 A_2}^{\eps})^{\otimes m}, \cE(\rho_{A_1 A_2})^{\otimes m}\big) \leq m \eps \, ,
 \end{align*} 
 which completes the proof.
\end{proof}

\begin{remark}
Given the result from Proposition~\ref{prop_dmax_chainrule} and the intuition from the max-relative entropy~\eqref{eq_chainDmax}, one may be tempted trying to prove~\eqref{eq_dmax_chainrule} for an arbitrary trace-preserving completely positive map $\cG$ that satisfies $\cF = \cR \circ\, \cG$. The following example however shows this is not possible.
Consider $\cE = \cI$, $\cF = \eps \cI$, $\rho = \proj{0}$, $\sigma = \proj{1}$, and $\cG(X) = (1-\eps)\proj{2} + \eps X$. Then we have $\cF = \cR \circ \cG$ with $\cR(X) = (\proj{0} + \proj{1})X(\proj{0} + \proj{1})$. Then for $\eps, \eps'$ not too large, we have $D^{\eps + \eps'}_{\max}\big(\cE(\rho) \| \cF(\sigma) \big) = + \infty$, as well as $D^{\eps'}(\cE \| \cF) \leq \log(1/\eps)$ and $D^{\eps}_{\max}(\cG(\rho) \| \cG(\sigma)) \leq D_{\max}((1-\eps) \proj{2} \| (1-\eps) \proj{2} + \eps \sigma) \leq 0$.
\end{remark}

\subsection{Chain rule for the quantum relative entropy}
In this section we prove a chain rule for the relative entropy and discuss its implications. To do so we use the chain rule for the smooth max-relative entropy from Proposition~\ref{prop_dmax_chainrule} and apply the asymptotic equipartition property (AEP) of this quantity. Hence we start by recalling that result.
\begin{lemma}[AEP for smooth max-relative entropy] \label{lem_AEP_dmax}
Let $\rho \in \St(A)$ and $\sigma \in \Pos(A)$. For any $\eps \in (0,1)$ and $n \geq 2 g(\eps)$, we have
\begin{align}
\label{AEP_dmax_ub}
\frac{1}{n} D^{\eps}_{\max}(\rho^{\otimes n} \| \sigma^{\otimes n}) \leq D(\rho \| \sigma) + \frac{4 (\log \mu) \sqrt{g(\eps)} }{\sqrt{n}}  \ ,
\end{align}
where $g(\eps) = \log(2/\eps^2)$ and $\mu = 1 + \tr\,\rho^{\frac{3}{2}} \sigma^{-\frac{1}{2}}+ \tr\, \rho^{\frac{1}{2}} \sigma^{\frac{1}{2}} $.
In addition, for any sequence $(\eps_n)_{n \in \N}$ with $\eps_n \in (0,1)$ such that $\lim_{n \to \infty} \eps_n < 1$, we have
\begin{align}
\label{AEP_dmax_lb}
\lim_{n \to \infty} \frac{1}{n} D^{\eps_n}_{\max}(\rho^{\otimes n} \| \sigma^{\otimes n}) \geq  D(\rho \| \sigma) \ .
\end{align}
\end{lemma}
\begin{proof}
The upper bound~\eqref{AEP_dmax_ub} can be found in~\cite[Theorem 6.4]{tomamichel_phd}, see also \cite{tomamichel09} and \cite[Equation (6.96)]{marco_book}.
%
%
%
%
For the lower bound~\eqref{AEP_dmax_lb}, it can be found in~\cite{tomamichel2013hierarchy} and it is stated in the form we need in~\cite[Lemma 2.1]{FR14}.
\end{proof}

We are finally ready to state a chain rule for the relative entropy.
\begin{theorem}[Chain rule for relative entropy] \label{thm_chainRule}
Let $\rho \in \St(A)$, $\sigma \in \Pos(A)$, $\cE \in \TPCP(A,B)$, and $\cF\in \CP(A,B)$ such that $D_{\max}(\cE \| \cF) < \infty$. Then
\begin{align} \label{eq_chainRule}
D\big(\cE(\rho) \| \cF(\sigma) \big) \leq D\big( \rho \| \sigma \big) +  \bar{D}^{\mathrm{reg}}(\cE \| \cF) \, .
\end{align}
In addition, in case $A = A_1 \otimes A_2$ and if there exists $\cR \in \CP(A_2,B)$ such that $\cF = \cR \circ \tr_{A_1}$ this inequality can be strengthened to 
\begin{align} \label{eq_chainRule_general}
D\big(\cE(\rho_{A_1 A_2}) \| \cF(\sigma_{A_1 A_2}) \big) \leq D( \rho_{A_2} \| \sigma_{A_2}) +  \bar{D}^{\mathrm{reg}}(\cE \| \cF) \, .
\end{align}
\end{theorem}

The following remark highlights the usefulness of the strengthened chain rule~\eqref{eq_chainRule_general} compared to its simplified version~\eqref{eq_chainRule}.  
\begin{remark}[The chain rule for conditional entropies] \label{rmk_chainRuleCond}
It is instructive to observe the following consequence of the usual chain rule for conditional entropies can be seen as a special case of our new chain rule~\eqref{eq_chainRule_general}.\footnote{Recall that the chan rule for conditional entropies ensures that for any $\rho_{C_1 C_2 D} \in \St(C_1 \otimes C_2 \otimes D)$ we have $H(C_1 C_2 | D)_{\rho} = H(C_1 |D)_{\rho} + H(C_2 | C_1 D)_{\rho}$, where the \emph{conditional entropy} is defined by $H(C|D):=-D(\rho_{CD} \| \id_C \otimes \rho_D)$.} Namely the inequality
\begin{align}
\label{eq_chain_rule_cond_entr}
H(C_1 C_2 | D)_{\rho} \geq H(C_1 |D)_{\rho} + \min_{\nu_{C_1 C_2 D} \in \mathrm{Y}} H(C_2 | C_1 D)_{\nu}\, ,
\end{align}
for $ \mathrm{Y}=\{ \nu_{C_1 C_2 D} \in \St(C_1 \otimes C_2 \otimes D) :  \nu_{C_1 C_2 D} = \nu_{C_1 D}^{1/2}\, \rho_{C_2 | C_1 D}\, \nu_{C_1 D}^{1/2} \}$ with $\rho_{C_2 | C_1 D} = \rho_{C_1 D}^{-1/2}\, \rho_{C_1 C_2 D} \, \rho_{C_1 D}^{-1/2}$. To see this let $E \simeq C_1 D$ and define the channel $\cA \in \TPCP(E,C_2)$ to have a Choi state $\rho_{C_2 | C_1 D}$ and $\cB \in \CP(E,C_2)$ as $\cB(X_E) = \id_{C_2} \tr\, X_E$. Then define the state $\rho_{C_1 D E} := \rho_{C_1D}^{1/2} \proj{\Omega} \rho_{C_1 D}^{1/2}$, where $\proj{\Omega}$ is an unnormalized maximally entangled state between $C_1D$ and $E$ and  we also define $\sigma_{C_1 D E} := \id_{C_1} \otimes \rho_{DE}$. Note that $\cA(\rho_{C_1D E}) = \rho_{C_1C_2 D}$ and $\cB(\sigma_{C_1 D E}) = \id_{C_1 C_2} \otimes \rho_{D}$. For $\cE = \cI_{C_1 D} \otimes \cA$ and $\cF=\cI_{C_1 D} \otimes \cB$ the chain rule~\eqref{eq_chainRule_general} gives
\begin{align*}
D(\rho_{C_1 C_2 D} \| \id_{C_1 C_2} \otimes \rho_{D}) 
&=D\big( \cE(\rho_{C_1 D E}) \| \cF(\sigma_{C_1 D E}) \big)\\
&\leq D(\rho_{C_1 D} \| \id_{C_1} \otimes \rho_{D} ) + D(\cA \| \cB) \\
&= D(\rho_{C_1 D} \| \id_{C_1} \otimes \rho_{D} ) + \max_{\nu_{C_1D} \in \St(C_1 \otimes D)} D(\nu_{C_1 D}^{\frac{1}{2}} \rho_{C_2 | C_1 D} \nu_{C_1 D}^{\frac{1}{2}} \| \id_{C_2} \otimes \nu_{C_1 D})  \, ,
\end{align*}
where in the second step we used that the channel relative entropy for replacer channels is additive under tensor products (see Remark~\ref{rmk_singleLetter}).
Written in terms of conditional entropies, this gives \eqref{eq_chain_rule_cond_entr}.
\end{remark}

\begin{proof}[Proof of Theorem~\ref{thm_chainRule}]
We may assume that $\supp(\rho) \subseteq \supp(\sigma)$, as otherwise, the right hand side is infinite.
The data-processing inequality implies $D( \rho_{A_2} \| \sigma_{A_2}) \leq D( \rho_{A_1 A_2} \| \sigma_{A_1 A_2})$ showing that~\eqref{eq_chainRule_general} implies~\eqref{eq_chainRule}. It thus suffices to prove~\eqref{eq_chainRule_general}. 
Let $n,m \in \N$ and $\eps, \eps' \in (0,1)$ that will be chosen later. Using Proposition~\ref{prop_dmax_chainrule} with maps $\cE^{\otimes n}$ and $\cF^{\otimes n}$ and states $\rho_{A_1 A_2}^{\otimes n}$ and $\sigma_{A_1 A_2}^{\otimes n}$, we have\footnote{The assumption that there exists $\cR \in \CP(A_2,B)$ such that $\cF=\cR \circ \tr_{A_1}$ implies that there also exists $\cR^n \in \CP(A_2^{\otimes n},B^{\otimes n})$, for example $\cR^n = \cR^{\otimes n}$, that satisfies $\cF^{\otimes n} = \cR^n \circ \tr_{A_1^{\otimes n}}$.}
\begin{align}
&D_{\max}^{m\eps + \sqrt{m \eps}+\eps'}\big( \left(\cE^{\otimes n}(\rho_{A_1 A_2}^{\otimes n}) \right)^{\otimes m} \| \left(\cF^{\otimes n}(\sigma_{A_1 A_2}^{\otimes n}) \right)^{\otimes m} \big) \nonumber \\
&\hspace{5mm}\leq m D^{\eps}_{\max}\big( \rho_{A_2}^{\otimes n} \| \sigma_{A_2}^{\otimes n} \big) + \max_{\nu \in \St(A_1^{\otimes n} \otimes A_2^{\otimes n})} D_{\max}^{\eps'}\big( \left(\cE^{\otimes n}(\nu)\right)^{\otimes m} \| \left(\cF^{\otimes n}(\nu)\right)^{\otimes m} \big) - m \log(1-\eps) \ . \label{eq_mid2}
\end{align}
Lemma~\ref{lem_AEP_dmax} implies that 
\begin{align*}
\frac{1}{n} D^{\eps}_{\max}\big( \rho_{A_2}^{\otimes n} \| \sigma_{A_2}^{\otimes n} \big) \leq D(\rho_{A_2} \| \sigma_{A_2}) + \frac{4(\log \mu)g(\eps)}{\sqrt{n}} \ ,
\end{align*}
where $\mu = 1 + \tr \, \rho_{A_2}^{\frac{3}{2}} \sigma_{A_2}^{-\frac{1}{2}} + \tr\, \rho_{A_2}^{\frac{1}{2}} \sigma_{A_2}^{\frac{1}{2}}$. Note that $\mu$ is a finite constant as we assumed that $\supp(\rho_{A_1 A_2}) \subseteq \supp(\sigma_{A_1 A_2})$. 
For the other relative entropy term on the right hand side of~\eqref{eq_mid2}, we also use Lemma~\ref{lem_AEP_dmax}
\begin{align*}
\frac{1}{m} D_{\max}^{\eps'}\big( \left(\cE^{\otimes n}(\nu)\right)^{\otimes m} \| \left(\cF^{\otimes n}(\nu)\right)^{\otimes m} \big)
&\leq D(\cE^{\otimes n}(\nu) \| \cF^{\otimes n}(\nu)) + \frac{4 (\log \mu') \sqrt{g(\eps')} }{\sqrt{m}} \ ,
\end{align*}
where $\mu' = 1 + \tr\,\cE^{\otimes n}(\nu)^{\frac{3}{2}} \cF^{\otimes n}(\nu)^{-\frac{1}{2}} + \tr\, \cE^{\otimes n}(\nu)^{\frac{1}{2}} \cF^{\otimes n}(\nu)^{\frac{1}{2}}$. Observe that for any $\nu \in {\St(A_1^{\otimes n} \otimes A_2^{\otimes n})}$, $\nu \leq \id_{A_1 \otimes A_2}^{\otimes n}$ and thus using the operator monotonicity of the square root function~\cite[Proposition~V.1.8]{bhatia_book}, we get
\begin{align*}
\tr\, \cE^{\otimes n}(\nu)^{\frac{1}{2}} \cF^{\otimes n}(\nu)^{\frac{1}{2}}
&\leq \tr\, (\cE(\id_{A_1 \otimes A_2})^{\otimes n})^{\frac{1}{2}} (\cF(\id_{A_1 \otimes A_2})^{\otimes n})^{\frac{1}{2}}
= \left(\tr\, \cE(\id_{A_1 \otimes A_2})^{\frac{1}{2}} \cF(\id_{A_1 \otimes A_2})^{\frac{1}{2}} \right)^n \ .
\end{align*}
For the other term, observe that for any $\nu \in \St(A_1^{\otimes n} \otimes A_2^{\otimes n})$, by definition of the max-relative entropy between channels, we have
\begin{align*}
\cE^{\otimes n}(\nu) &\leq 2^{D_{\max}(\cE^{\otimes n} \| \cF^{\otimes n})} \cF^{\otimes n}(\nu) \ .
\end{align*} 
In addition, we have $D_{\max}(\cE^{\otimes n} \| \cF^{\otimes n}) = n D_{\max}(\cE \| \cF)$~\cite[Lemma~12]{berta18}. Now use the operator anti-monotonicity of the function $x \mapsto x^{-\frac{1}{2}}$~\cite[Table~2.2]{Sutter_book}, we have $\cF^{\otimes n}(\nu)^{-\frac{1}{2}} \leq 2^{\frac{1}{2} n D_{\max}(\cE \| \cF)} \cE^{\otimes n}(\nu)^{-\frac{1}{2}}$ and hence
\begin{align*}
\tr\,\cE^{\otimes n}(\nu)^{\frac{3}{2}} \cF^{\otimes n}(\nu)^{-\frac{1}{2}}
\leq 2^{\frac{1}{2} n D_{\max}(\cE \| \cF)} \tr \, \cE^{\otimes n}(\nu)
= 2^{\frac{1}{2} n D_{\max}(\cE \| \cF)} \ .
\end{align*}
As a result, for $C := \log\big(1 + 2^{\frac{1}{2} D_{\max}(\cE \| \cF)} + \tr\,\cE(\id_{A_1 \otimes A_2})^{\frac{1}{2}} \cF(\id_{A_1 \otimes A_2})^{\frac{1}{2}} \big)$ we have $\log \mu' \leq n C$. Note that $C < \infty$ is a constant independent of $n$ or $m$ because by assumption $D_{\max}(\cE \| \cF) < \infty$. Putting things together we get
\begin{align*}
&\frac{1}{nm} D_{\max}^{m\eps + \sqrt{m \eps}+\eps'}\Big( \left(\cE^{\otimes n}(\rho_{A_1 A_2}^{\otimes n}) \right)^{\otimes m} \| \left(\cF^{\otimes n}(\sigma_{A_1 A_2}^{\otimes n}) \right)^{\otimes m} \Big) \\
&\leq  D(\rho_{A_2} \| \sigma_{A_2}) + \frac{4(\log \mu)g(\eps)}{\sqrt{n}}  + \frac{1}{n} \left(\max_{\nu \in \St(A_1^{\otimes n} \otimes A_2^{\otimes n})} D(\cE^{\otimes n}(\nu) \| \cF^{\otimes n}(\nu)) \!+\! \frac{4 C n \sqrt{g(\eps')} }{\sqrt{m}} - \log(1-\eps) \right) \ .
\end{align*}
Now choose $m = n, \eps_n = \frac{1}{9n}$, $\eps' = \frac{1}{9}$ and note that $2g(\eps_n) = 2\log(162 n^2) \leq n$ for large enough $n$, so that~\eqref{AEP_dmax_ub} is applicable. As a result, using~\eqref{AEP_dmax_lb}, the left hand side gives in the limit $n \to \infty$:
\begin{align*}
\lim_{n \to \infty} \frac{1}{n^2} D_{\max}^{\frac{5}{9}}\big( \left(\cE^{\otimes n}(\rho_{A_1 A_2}^{\otimes n}) \right)^{\otimes n} \| \left(\cF^{\otimes n}(\sigma_{A_1 A_2}^{\otimes n}) \right)^{\otimes n} \big) \geq D\big(\cE(\rho_{A_1 A_2}) \| \cF(\sigma_{A_1 A_2})\big) \ ,
\end{align*}
and the right hand side gives in the limit $n \to \infty$:
\begin{align*}
D(\rho_{A_2} \| \sigma_{A_2})  + \lim_{n \to \infty} \frac{1}{n} \max_{\nu \in \St(A_1^{\otimes n}\otimes A_2^{\otimes n} )} D\big(\cE^{\otimes n}(\nu) \| \cF^{\otimes n}(\nu)\big) \, ,
\end{align*}
which thus completes the proof.
\end{proof}

An important corollary to Theorem~\ref{thm_chainRule} is when the maps are of the form $\cI \otimes \cE$ and $\cI \otimes \cF$. In this case, in the right hand side, we get the more common (stabilized) relative entropy between channels.
\begin{corollary} \label{corollary_collapse}
Let $\cE \in \TPCP(A,B)$ and $\cF \in \CP(A,B)$. Then
\begin{align}
\sup_{\rho_{RA}, \sigma_{RA} \in \St(R \otimes A)} D( \cE(\rho_{RA}) \| \cF(\sigma_{RA}) ) - D(\rho_{RA} \| \sigma_{RA}) = D^{\mathrm{reg}}(\cE \| \cF) \ , \label{eq_chainCor}
\end{align}
where the supremum is over all possible finite dimensional systems $R$.
In other words,
\begin{align} \label{eq_collapse}
D^{A}(\cE \| \cF) = D^{\mathrm{reg}}(\cE \| \cF) \ .
\end{align}
\end{corollary}
We note that another interesting feature of~\eqref{eq_chainCor} is that it shows that for any $\cE \in \TPCP(A,B)$ and any $\cF  \in \CP(A,B)$ there exist states $\rho_{RA}$ and $\sigma_{RA}$ such that the chain rule holds with equality.
\begin{proof}
It is known that 
\begin{align*}
D^{\mathrm{reg}}(\cE \| \cF) 
\leq D^{A}(\cE \| \cF)  \, ,
\end{align*}
is correct because of the operational interpretation of these two quantities in the area of asymptotic quantum channel discrimination. This is discussed in detail in Section~\ref{sec_channelDiscrimination}.
Theorem~\ref{thm_chainRule} applied to channels $\cI_R \otimes \cE$ and $\cI_R \otimes \cF$ shows that
\begin{align*} 
D\big((\cI_R \otimes \cE)(\rho_{RA}) \| (\cI_R \otimes \cF)(\sigma_{RA}) \big) - D(\rho_{RA} \| \sigma_{RA})
\leq  \bar D^{\mathrm{reg}}(\cI_R \otimes \cE \|  \cI_R \otimes \cF)
\end{align*}
To conclude it suffices to observe that for any system $R$ we have $\bar D^{\mathrm{reg}}(\cI_R \otimes \cE \|  \cI_R \otimes \cF) \leq D^{\mathrm{reg}}(\cE \| \cF)$.
\end{proof}

\begin{remark}[Examples where the chain rule is single-letter] \label{rmk_singleLetter}
 We note that the regularized relative entropy term in Theorem~\ref{thm_chainRule} cannot be single-letterized in full generality without weakening the bounds as this quantity is not additive under tensor products as shown by Proposition~\ref{prop_notAdditive}. For channels with a specific structure their channel relative entropy is additive under the tensor product which implies that $D^{\mathrm{reg}}(\cE \| \cF)=D(\cE \| \cF)$. Examples of such channels are
 \begin{enumerate}[(i)]
\item classical-quantum channels~\cite[Lemma~25]{berta18}
\item covariant channels with respect to the unitary group~\cite[Corollary II.5]{leditzky2018approaches}
\item $\cE$ arbitrary and $\cF$ a replacer channel (i.e., $\cF(X) = \omega\, \tr\, X$ for $\omega \in \St(B)$)~\cite[Proposition~41]{berta18}.  
 \end{enumerate}
\end{remark}

\begin{remark}[Relaxation of the chain rule] 
We can single-letterize the chain rule from Theorem~\ref{thm_chainRule} by replacing the regularized channel relative entropy term with the Belavkin-Stasewski channel relative entropy defined by $\widehat D(\cE \| \cF) := \max_{\phi_{RA} \in \St(A \otimes A) } \tr\, \cE(\phi) \log \big(\cE(\phi)^{\frac{1}{2}} \cF(\phi)^{-1} \cE(\phi)^{\frac{1}{2}} )$. We note that the logarithmic trace inequality~\cite{HP93,ando94} (see also~\cite[Theorem~4.6]{Sutter_book}) ensures that $D(\cE \| \cF) \leq \widehat D(\cE \| \cF)$. Furthermore, the Belavkin-Stasewski channel relative entropy is additive under tensor products~\cite[Lemma~6]{kun19}. Another benefit from this relaxation is the fact that $\widehat D(\cA\|\cB)$ has an explicit form and is thus efficiently computable~\cite[Lemma~5]{kun19}. 
\end{remark}

\subsection{Asymptotic quantum channel discrimination}   \label{sec_channelDiscrimination}
A fundamental task in quantum information theory is to distinguish between two quantum channels $\cE, \cF \in \TPCP(A,B)$. For this problem one usually differentiates between two different classes of strategies:
\begin{enumerate}[(a)]
\item \emph{Non-adaptive strategies} (also called parallel strategies): Here we are given ``black-box'' access to $n$ uses of a channel $\cG$, which is either $\cE$ or $\cF$, that can be used in parallel before performing a measurement. More precisely, for an arbitrary state $\rho_{A^n R} \in \St(A_1 \otimes \ldots \otimes A_n \otimes R)$ with a reference system $R$ we create the state $\sigma_{B^n R}=\cG^{\otimes n}(\rho_{A^n R})$ and perform a measurement on $\sigma_{B^n R}$. Based on the measurement outcome we try to guess if $\cG=\cE$ or $\cG=\cF$. The protocol is depicted in Figure~\ref{fig_parallel}.
It has been shown recently~\cite[Theorem 3]{Wang2019} that in the asymmetric regime where we fix the type-I error to be bounded by $\eps$, the asymptotic optimal rate of the type-II error exponent is given by $D^{\mathrm{reg}}(\cE \| \cF)$, when $\eps$ goes to 0.

\item \emph{Adaptive strategies} (also called sequential strategies): Here we are also given ``black-box'' access to $n$ uses of a channel $\cG$ which is either $\cE$ or $\cF$. However unlike in the non-adaptive scenario, after each use of a channel we are allowed to perform an adaptive trace-preserving completely positive map $\cN \in \TPCP(B \otimes R, A \otimes R)$ before we perform a measurement at the end. More precisely, for an arbitrary state $\rho^{(0)}_{AR} \in \St(A \otimes R)$ we create $\rho_{AR}^{(k)}=(\cN \circ\, \cG)(\rho_{AR}^{(k-1)})$ for $k=1,\ldots,n$. Finally we perform a measurement on $\rho_{AR}^{(n)}$ and based on the outcome try to guess if $\cG= \cE$ or $\cG = \cF$.
The strategy is depicted in Figure~\ref{fig_adaptive}. The asymptotically optimal rate of the type-II error exponent for this strategy is given by $D^{A}(\cE \| \cF)$~\cite[Theorem 6]{Wang2019}.
\end{enumerate}

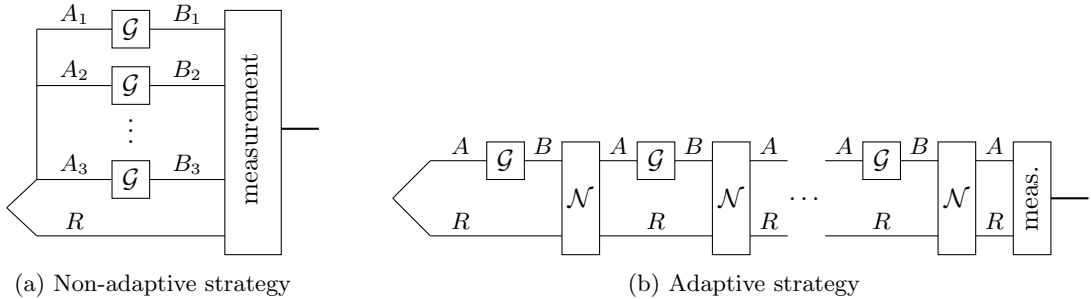
\begin{figure}[!h]
\centering
\begin{subfigure}[b]{0.25\textwidth}
\begin{tikzpicture}
\draw (0,0) -- (1,0);
\draw (1,-0.25) rectangle (1.5,0.25);
\node at (1.25,0) {$\cG$};
\node at (0.5,0.2) {\small $A_1$};
\node at (2,0.2) {\small$B_1$};
\draw (1.5,0) -- (2.5,0);
\draw (0,-0.75) -- (1,-0.75);
\draw (1,-1) rectangle (1.5,-0.5);
\node at (1.25,-0.75) {$\cG$};
\node at (0.5,0.2-0.75) {\small$A_2$};
\node at (2,0.2-0.75) {\small$B_2$};
\draw (1.5,-0.75) -- (2.5,-0.75);
\draw (0,-2) -- (1,-2);
\draw (1,-2.25) rectangle (1.5,-1.75);
\node at (1.25,-2) {$\cG$};
\node at (0.5,-1.8) {\small$A_3$};
\node at (2,-1.8) {\small$B_3$};
\draw (1.5,-2) -- (2.5,-2);
\node at (1.25,-1.25) {$\vdots$};
\draw (0,-2) -- (0,0);
\node at (0.5,-2.55) {\small$R$};
\draw (0,-2.75) -- (2.5,-2.75);
\draw (-0.4,-2.375) -- (0,-2.75);
\draw (-0.4,-2.375) -- (0,-2.0);
\draw (2.5,-3) rectangle (3.25,0.25);
\node[rotate=90] at (2.825,-1.325) {measurement};
\draw[thick] (3.25,-1.325) -- (3.75,-1.325);
\end{tikzpicture}
\subcaption{Non-adaptive strategy}
\label{fig_parallel}
\end{subfigure}
\hspace{10mm}
\begin{subfigure}[b]{0.6\textwidth}
\begin{tikzpicture}
\draw (0,0) -- (0.5,0.5);
\draw (0,0) -- (0.5,-0.5);
\draw (0.5,0.5) -- (1.25,0.5);
\draw (0.5,-0.5) -- (2.25,-0.5);
\draw (1.25,0.25) rectangle (1.75,0.75);
\node at (1.5,0.5) {$\cG$};
\node at (2,0.7) {\small$B$};
\node at (0.9,0.7) {\small$A$};
\node at (0.9,-0.3) {\small$R$};
\draw (1.75,0.5) -- (2.25,0.5);
\draw (2.25,0.75) rectangle (2.75,-0.75);
\node at (2.5,0) {$\cN$};
\node at (3,0.7) {\small$A$};
\node at (4,0.7) {\small$B$};
\node at (3.5,-0.3) {\small$R$};
\draw (2.75,-0.5) -- (4.25,-0.5);
\draw (2.75,0.5) -- (3.25,0.5);
\draw (3.75,0.5) -- (4.25,0.5);
\draw (3.25,0.25) rectangle (3.75,0.75);
\node at (5,0.7) {\small$A$};
\node at (5,-0.3) {\small$R$};
\node at (3.5,0.5) {$\cG$};
\draw (4.25,0.75) rectangle (4.75,-0.75);
\node at (4.5,0) {$\cN$};
\draw (4.75,-0.5) -- (5.25,-0.5);
\draw (4.75,0.5) -- (5.25,0.5);
\node at (5.5,0) {\ldots};
\draw (5.75,-0.5) -- (7.25,-0.5);
\draw (5.75,0.5) -- (6.25,0.5);
\draw (6.75,0.5) -- (7.25,0.5);
\node at (6,0.7) {\small$A$};
\node at (7,0.7) {\small$B$};
\node at (6.5,-0.3) {\small$R$};
\draw (6.25,0.25) rectangle (6.75,0.75);
\node at (6.5,0.5) {$\cG$};
\draw (7.25,0.75) rectangle (7.75,-0.75);
\node at (7.5,0) {$\cN$};
\node at (8,0.7) {\small$A$};
\node at (8,-0.3) {\small$R$};
\draw (7.75,-0.5) -- (8.25,-0.5);
\draw (7.75,0.5) -- (8.25,0.5);
\draw (8.25,0.75) rectangle (8.75,-0.75);
\node[rotate=90] at (8.5,0) {meas.};
\draw[thick] (8.75,0) -- (9.25,0);
\end{tikzpicture}
\subcaption{Adaptive strategy}
\label{fig_adaptive}
\end{subfigure}
\caption{Two general protocols for non-adaptive and adaptive strategies for the task of channel discrimination. The channel $\cG$ is either $\cE$ or $\cF$ and the task is to distinguish between these two cases.}
\label{fig_channel_disc}
\end{figure}
Because a non-adaptive strategy can be viewed as a particular instance of an adaptive strategy~\cite{chiribella08} it follows that adaptive strategies are clearly as powerful as non-adaptive ones, which in technical terms means
\begin{align*}
D^{\mathrm{reg}}(\cE \| \cF)
\leq D^{A}(\cE \| \cF) \, .
\end{align*}
It has been an open question if adaptive strategies can be more powerful for the task of asymptotic quantum channel discrimination. For some special classes of channels, such as classical and classical-quantum channels it has been shown that adaptive protocols cannot improve the error rate for asymmetric channel discrimination~\cite{hayashi2009discrimination,berta18}. Corollary~\ref{corollary_collapse} now proves that this is the case for all quantum channels because
 \begin{align*}
D^{\mathrm{reg}}(\cE \| \cF)
= D^{A}(\cE \| \cF) \, .
\end{align*}

We note that this is surprising for various reasons. In the symmetric Chernoff setting~\cite{duan09,harrow10,duan16} adaptive protocols offer an advantage over non-adaptive ones. Furthermore, in the non-asymptotic setting adaptive protocols also outperform non-adaptive strageties~\cite{duan09,harrow10,Puzzuoli2017}.

\paragraph{Acknowledgements.} We thank Mario Berta for discussions about asymptotic quantum channel discrimination. We also thank Mark Wilde for presenting the question if adaptive strategies can outperform non-adaptive strategies for the task of asymmetric channel discrimination as an open problem in the Banff workshop on ``Algebraic and Statistical ways into Quantum Resource Theories''~\cite{banff19}.
KF is supported by the University of Cambridge Isaac Newton Trust Early Career grant RG74916.
OF acknowledges support from the French ANR project ANR-18-CE47-0011 (ACOM).
DS and RR acknowledge support from the Swiss National Science Foundation via the NCCR QSIT as well as project No.~200020\_165843 and Air Force Office of Scientific Research grant FA9550-19-1-0202.


\bibliographystyle{arxiv_no_month}
\bibliography{bibliofile}

\end{document}